\documentclass[aps,prd,reprint,nofootinbib,superscriptaddress]{revtex4-2}
\usepackage{amsmath,amssymb,amsfonts,mathtools,bm}
\usepackage{amsthm}
\usepackage{graphicx}
\usepackage{booktabs}
\usepackage{xcolor}
\usepackage[colorlinks=true,allcolors=blue!55!black]{hyperref}
\usepackage{microtype}
\newtheorem{proposition}{Proposition}

\newcommand{\dd}{\mathrm d}
\newcommand{\Tr}{\operatorname{Tr}}
\newcommand{\diag}{\operatorname{diag}}
\newcommand{\e}{\mathrm e}
\newcommand{\Mcal}{\mathcal M}
\newcommand{\Ncal}{\mathcal N}
\newcommand{\Fcal}{\mathcal F}

\newcommand{\rev}[1]{#1}
\newcommand{\fresh}[1]{#1}

\newcommand{\newpass}[1]{#1}
\begin{document}
\title{Thermal Alignment as a Pathway to Axion Dark Matter}
\author{Imtiaz Khan}
\email{ikhanphys1993@gmail.com}
\affiliation{Department of Physics, Zhejiang Normal University, Jinhua, Zhejiang 321004, China}
\affiliation{Research Center of Astrophysics and Cosmology, Khazar University, Baku, AZ1096, 41 Mehseti Street, Azerbaijan}
\author{G. Mustafa}
\email{gmustafa3828@gmail.com}
\affiliation{Department of Physics, Zhejiang Normal University, Jinhua, Zhejiang 321004, China}
\author{Chengxun Yuan}
\email{yuancx@hit.edu.cn}
\affiliation{School of Physics, Harbin Institute of Technology, Harbin 150001, People's Republic of China}
\author{Farkhod Botirov}
\email{f.botirov@nuu.uz}
\affiliation{National University of Uzbekistan, Tashkent 100174, Uzbekistan}
\author{Ahmadjon Abdujabbarov}
\email{ahmadjonab@gmail.com}
\affiliation{School of Physics, Harbin Institute of Technology, Harbin 150001, People's Republic of China}
\author{Farruh Atamurotov}
\email{atamurotov@yahoo.com}
\affiliation{Kimyo International University in Tashkent, Shota Rustaveli str. 156, Tashkent 100121, Uzbekistan}
\begin{abstract}
We study the relic energy of an axion whose homogeneous displacement is damped by a thermal gauge plasma. Dissipation is accompanied by thermal fluctuations. To calculate the energy at decoupling, we retain the covariance between the axion field and its velocity. For linear Langevin dynamics, positivity of the noise covariance provides a finite-interval lower bound on the energy generated by the bath, independently of the initial state. We consider a renormalizable $SU(2)_h\times SU(3)_\ell$ theory in which a scalar transition separates the early topological susceptibility and Chern--Simons diffusion from the confining potential acting at late times. With the specified scalar history and a local plasma response, fluctuations generated during the final cooling interval carry energy comparable to the coherent-only dark matter benchmark. The relic-density integral extends over finite momenta and cannot be inferred from the homogeneous topological diffusion rate. We also include the chiral-condensate term in the scalar free energy. For the chirally broken benchmark, this term removes the stationary high-temperature solution of the polynomial potential. Thus, a consistent microscopic realization requires both an assessment of thermal axion production and scalar phases that remain stationary in the full free energy.
\end{abstract}
\maketitle
\section{Introduction}
\newpass{\label{sec:intro}

In an expanding Universe, a displaced axion begins to oscillate when its potential curvature becomes sufficiently large compared with the expansion rate. In conventional misalignment production, the initial angular displacement largely sets the energy in these oscillations \cite{Preskill1983,AbbottSikivie1983,DineFischler1983,Marsh2016,DiLuzioReview2020}. Related production mechanisms have also been examined for nonstandard inflationary histories \cite{Brandenberger2025}. If the axion interacts with a thermal plasma before its final potential is established, however, its homogeneous expectation value and its fluctuating component evolve differently. Dissipation attenuates the mean field, while the associated thermal force populates field and velocity fluctuations. The energy present after decoupling must then be calculated from both contributions. In particular, the removal of a classical expectation value is a statement about the first moment of the state. It places no corresponding bound on the positive quadratic moments that determine the energy density after the potential changes.

Thermal damping of axion oscillations was considered in early studies of the misalignment mechanism \cite{UnruhWald1985,Turner1985}. More recent analyses of dissipative axion dynamics in non-Abelian plasmas have examined the modified onset of oscillations, the suppression of their amplitude, and consequences for the relic abundance and isocurvature \cite{Papageorgiou2023,Choi2023}. These analyses establish how the coherent part responds to the thermal medium. The attenuation of that part, however, does not determine the energy remaining in thermally populated modes. The same gauge interactions responsible for friction generate fluctuations, and those fluctuations need not disappear when the coherent solution approaches its minimum.

In nonequilibrium descriptions, the decay of a mean field and the production of its fluctuations are different processes. For axions, an open-system calculation exhibits coherent damping alongside a mode population evolving toward thermal occupation \cite{CaoBoyanovsky2022}. Treatments with momentum-dependent response further show how gauge correlations in space and time enter axion production \cite{Gurmani2026}. Here we calculate the field--velocity covariance acquired while the axion remains coupled to the plasma. We do not assume that its spectrum has equilibrated when contact ends. Instead, the covariance is resolved into the propagated initial contribution and a positive term sourced by the bath. The latter determines a lower contribution to the released energy irrespective of the initial fluctuations. In particular, an unknown earlier preparation need not be specified in order to evaluate the energy injected during a chosen late interval.

How rapidly this population develops is set by the non-Abelian plasma response. The long-wavelength topological rate is characterized by Chern--Simons diffusion calculations \cite{MooreTassler2011,DOnofrio2014,GuinSharma2026}. At finite momentum, the relevant correlators may differ from their homogeneous limit; in particular, the soft-axion response distinguishes nearly lightlike kinematics from zero spatial momentum \cite{Bouzoud2026}. Consequently, the $k=0$ friction coefficient cannot by itself specify the integrated relic abundance. The energy integral involves a range of physical momenta, and the bath response is, in general, a function of momentum and frequency. Noise and dissipation are connected by the fluctuation--dissipation relation \cite{HohenbergHalperin1977,Miyamoto2014,Bartrum2015,Banerjee2026}. We first use a local Langevin description and then introduce a finite bath correlation time, treating friction and fluctuations with the same response function.

Before the final confinement potential develops, the axion may experience a temporary mass. This occurs, for example, in models with hidden monopoles, intermediate confinement, or trapped misalignment \cite{Kawasaki2018,CoGonzalezHarigaya2019,Koutsangelas2023,DiLuzioSorensen2024,Dvali1995,DvaliFitzKomisel2026}. First-order transitions and varying mass terms offer other production histories, while explicit symmetry breaking affects the evolution of defects \cite{Nakagawa2023,Zhang2024}. The amplitude of the final oscillations depends on the earlier mass history as well as the final curvature \cite{Khan:2026nsz,Gurmani:2026kuh}. Here the axion is initially subject to $SU(2)_h$ susceptibility and thermal friction. Its later mass is generated by confinement in the $SU(3)_\ell$ sector. A scalar transition changes the fermion masses and the relative importance of these two gauge contributions.

The confined gauge sector influences not only the axion potential but also the scalar vacuum \cite{Pirzada:2026sle}. Massless quarks eliminate dependence on the vacuum angle even in a confining theory with a chiral condensate. Once a scalar order parameter $R$ generates their masses, both the topological susceptibility and the chiral contribution to the scalar free energy change \cite{DiVecchiaVeneziano1980,LeutwylerSmilga1992}. We evaluate a polynomial finite-temperature scalar potential of the kind used in perturbative studies of symmetry breaking and nonrestoration \cite{DolanJackiw1974,Weinberg1974,MohapatraSenjanovic1979,GrossPisarskiYaffe1981}. We then include the condensate-induced scalar force. This comparison is necessary because a trajectory that is stationary in the polynomial approximation need not remain stationary after the chiral contribution is included.

The stochastic treatment also distinguishes the length scales relevant to relic energy from those relevant to cosmological perturbations. Finite-range thermal correlations give regular field power near $k=0$ and a dimensionless spectrum proportional to $k^3$, as in other systems with causal short-range ordering \cite{Kibble1976,Zurek1985,Chandran2012}. This behavior differs from the long-wavelength correlations associated with primordial fluctuations \cite{Pirzada:2026jml,Ijaz:2026ear,Khan:2026doo} and the nearly scale-invariant isocurvature constrained by CMB data \cite{Enander2017,Feix2020,Planck2018VI,Planck2018X}. The higher-momentum population still has to be followed: its gradient energy and subsequent free streaming may affect small-scale clustering \cite{Harigaya2025}. A first-order scalar transition could also generate gravitational radiation, but that question requires the actual nucleation and wall dynamics \cite{Linde1983,Caprini2020}.

In this work we specify the renormalizable gauge and scalar model, calculate the susceptibilities and the local Chern--Simons damping rate, and follow the axion mean and covariance through cooling and an imposed transition profile. We derive a positive finite-interval contribution to the released energy and compare its momentum integral with the coherent-only relic reference. We examine how changes in the plasma response and wall duration affect the fluctuation spectrum on microscopic scales and at CMB wavelengths. Finally, we calculate the chiral contribution to the scalar force and compare the full stationary-point condition with the polynomial thermal phases. The comparison of energies uses a local finite-momentum response, while the assumed scalar path is meaningful physically only for consistent thermal vacua.

We use natural units, except where units are given explicitly.
}
\section{Finite temperature gauge completion}\label{sec:uv}
\fresh{We consider two contributions to the axion curvature, acting before and after its interaction with the thermal plasma. They arise from different gauge sectors separated by a scalar transition. At high temperature, the $SU(2)_h$ plasma gives the topological susceptibility and Chern--Simons diffusion, whereas confinement in $SU(3)_\ell$ produces the final potential. The scalar potential permits a high-temperature branch through symmetry nonrestoration \cite{Weinberg1974,MohapatraSenjanovic1979}. Because the light-fermion masses also determine a chiral contribution to the scalar free energy, we examine the perturbative polynomial phases together with that nonperturbative contribution.}
\subsection{Field content and renormalizable Lagrangian}
We assign the early interaction and the late-time axion curvature to two factors of the gauge group
\begin{equation}
 \mathcal G=SU(2)_h\times SU(3)_\ell .
 \label{eq:gaugegroup}
\end{equation}
A complex Peccei--Quinn field $\Phi=(f_a+\rho)\e^{i\theta}/\sqrt2$ carries charge $+1$ under $U(1)_{\rm PQ}$. The transition fields are real singlets $S$ and $R$, an $SU(2)_h$ doublet $\mathcal H$, and an $O(N_X)$ spectator multiplet $X_A$. The Dirac fundamentals $q_i\sim(\bm2,\bm1)$ and $p_i\sim(\bm1,\bm3)$ acquire masses from $S$ and $R$, respectively. Two additional KSVZ fermions, $Q_h\sim(\bm2,\bm1)$ and $Q_\ell\sim(\bm1,\bm3)$, form vectorlike pairs under the gauge group and have chiral PQ assignments, as in anomalous hadronic axion models \cite{Kim1979,Shifman1980,DiLuzioReview2020}. We assign PQ charge $0$ to $Q_{hL}$ and $Q_{\ell L}$ and charge $-1$ to $Q_{hR}$ and $Q_{\ell R}$. Each pair then contributes one unit to the anomaly coefficient of its gauge sector. This assignment gives the axion a common origin in the two anomalous sectors while allowing the corresponding light-quark masses to change independently across the scalar transition. The renormalizable interactions are
\begin{align}
 \mathcal L_{\rm UV}={}&|\partial\Phi|^2+\frac12(\partial S)^2
 +\frac12(\partial R)^2+\frac12(\partial X_A)^2\nonumber\\
 &+(D_\mu\mathcal H)^\dagger D^\mu\mathcal H-V_0
 -\frac14G_h^2-\frac14G_\ell^2\nonumber\\
 &+\sum_{i=1}^{2}\bar q_i i\!\not\!Dq_i
 +\sum_{i=1}^{2}\bar p_i i\!\not\!Dp_i\nonumber\\
 &+\bar Q_h i\!\not\!DQ_h+\bar Q_\ell i\!\not\!DQ_\ell\nonumber\\
 &-y_qS\sum_i\bar q_iq_i-y_pR\sum_i\bar p_ip_i\nonumber\\
 &-\left(y_{Kh}\Phi\bar Q_{hL}Q_{hR}
 +y_{K\ell}\Phi\bar Q_{\ell L}Q_{\ell R}+{\rm h.c.}\right)\nonumber\\
 &+\frac{g_h^2\theta_h^0}{32\pi^2}G_h^a\widetilde G_h^a\nonumber\\
 &+\frac{g_\ell^2\theta_\ell^0}{32\pi^2}G_\ell^A\widetilde G_\ell^A .
 \label{eq:lagrangianUV}
\end{align}
Here $G_h^2\equiv G_{h\mu\nu}^aG_h^{a\mu\nu}$, with an analogous definition for $G_\ell^2$. The Dirac field content cancels gauge anomalies in both sectors. Under $Z_2^S$, $S\to-S$ and $q_{iR}\to-q_{iR}$ for both flavors; $Z_2^R$ acts analogously on $R$ and $p_{iR}$. With two flavors, the discrete chiral transformations carry no anomaly beyond a $2\pi$ shift of the vacuum angle. They exclude bare mass terms for $q_i$ and $p_i$. Since $\det M_q\propto S^2$ and $\det M_p\propto R^2$, changing the sign of either order parameter leaves the axion minimum equivalent.
For $y_{Kj}f_a/\sqrt2\gg T_0$, the PQ radial mode and KSVZ fermions are thermally inaccessible during the scalar transition. We integrate them out, retaining the anomalous axion couplings at lower energies. A chiral redefinition of the heavy fermions followed by a constant shift of $\theta$ gives \cite{Kim1979,Shifman1980}
\begin{equation}
 \mathcal L_{\theta}=\frac{g_h^2\theta}{32\pi^2}G_h^a\widetilde G_h^a
 +\frac{g_\ell^2(\theta-\delta)}{32\pi^2}G_\ell^A\widetilde G_\ell^A,
 \qquad \delta=\theta_h^0-\theta_\ell^0 .
 \label{eq:lagrangian}
\end{equation}
The axion removes one combination of the two vacuum angles; the physical relative angle is $\delta$. The limit $\delta\to0$ restores CP symmetry, while radiative corrections to a small displacement are proportional to CP-violating parameters. This angle determines the coherent displacement after capture. The covariance of the connected fluctuations is determined separately by the preceding thermal evolution.
The competing thermal scalar phases are described using the following renormalizable potential
\begin{align}
V_0={}&\lambda_\Phi\left(|\Phi|^2-\frac{f_a^2}{2}\right)^2
 +\left(|\Phi|^2-\frac{f_a^2}{2}\right)\mathcal P_\Phi\nonumber\\
&+\frac{m_S^2}{2}S^2+\frac{\lambda_S}{4}S^4
-\frac{\mu_R^2}{2}R^2+\frac{\lambda_R}{4}R^4 \nonumber\\
&-\mu_H^2\mathcal H^\dagger\mathcal H
+\lambda_H(\mathcal H^\dagger\mathcal H)^2 \nonumber\\
&+\frac{m_X^2}{2}X^2+\frac{\lambda_X}{4}(X^2)^2
+\frac{\kappa_{SR}}{4}S^2R^2
+\frac{\kappa_{SH}}{2}S^2\mathcal H^\dagger\mathcal H\nonumber\\
&-\frac{\kappa_{SX}}{4}S^2X^2
+\frac{\kappa_{RH}}{2}R^2\mathcal H^\dagger\mathcal H\nonumber\\
&+\frac{\kappa_{RX}}{4}R^2X^2
+\frac{\kappa_{HX}}{2}\mathcal H^\dagger\mathcal H X^2,
\label{eq:potential0}
\end{align}
where $X^2\equiv X_AX_A$ and
\begin{equation}
 \mathcal P_\Phi=\lambda_{\Phi S}S^2+\lambda_{\Phi R}R^2
 +2\lambda_{\Phi H}\mathcal H^\dagger\mathcal H+\lambda_{\Phi X}X^2 .
 \label{eq:Phiportals}
\end{equation}
At $\mu=T_0$, we take $\lambda_\Phi=0.50$, $y_{Kh}=y_{K\ell}=1$, and $\mathcal P_\Phi=0$. The corresponding masses are $m_\rho=f_a$ and $m_{Q_h}=m_{Q_\ell}=f_a/\sqrt2$, more than seven hundred times the transition scale. The radial field and KSVZ fermions are consequently absent from the thermal spectrum at the transition. Their role in the effective description is to transmit the anomalous PQ interaction, not to provide additional light degrees of freedom during the thermal rearrangement of the scalar potential. The light scalar fields have no shared Yukawa vertex with $\Phi$, so gauge-induced $\Phi$ portals first arise beyond one loop. We also set $\kappa_{RH}=\kappa_{RX}=\kappa_{HX}=0$ at $\mu=T_0$; the induced mixed couplings lie within the ten percent variation of thermal coefficients considered below. The negative $S$--$X$ portal permits high-temperature symmetry nonrestoration in a multi-scalar theory \cite{Weinberg1974,MohapatraSenjanovic1979}. Stability in the $S$--$X$ directions requires
\begin{equation}
 \kappa_{SX}^2<4\lambda_S\lambda_X .
 \label{eq:boundedness}
\end{equation}
The remaining mixed quartics used in the benchmark are nonnegative.
\subsection{Thermal masses and phase exchange}
The competing thermal phases require the preferred scalar expectation values to change as the plasma cools. We first calculate this competition from the leading high-temperature terms of the one-loop perturbative potential \cite{DolanJackiw1974,Weinberg1974,MohapatraSenjanovic1979}. With $\mathcal H^\dagger\mathcal H=h^2/2$, the thermal mass coefficients are
\begin{align}
 c_S={}&\frac{\lambda_S}{4}+\frac{\kappa_{SR}}{12}
 +\frac{\kappa_{SH}}{3}+\frac{N_f^h d_h y_q^2}{6}
 -\frac{N_X\kappa_{SX}}{12},\label{eq:cS}\\
 c_R={}&\frac{\lambda_R}{4}+\frac{\kappa_{SR}}{12}
 +\frac{N_f^\ell d_\ell y_p^2}{6},\label{eq:cR}\\
 c_H={}&\frac{\lambda_H}{2}+\frac{3g_h^2}{16}
 +\frac{\kappa_{SH}}{12},\label{eq:cH}\\
 c_X={}&\frac{(N_X+2)\lambda_X}{12}-\frac{\kappa_{SX}}{12}.
 \label{eq:cX}
\end{align}
The spectator multiplicity allows $c_S<0$ while Eq.~\eqref{eq:boundedness} remains satisfied. At the same time, $c_X>0$ stabilizes $X_A=0$ through the transition interval. In the polynomial approximation, negative $c_S$ favors a nonzero $S$ in the high-temperature phase; on cooling, the competing $(R,h)$ minimum can become energetically preferred. \rev{The resulting thermal potential is}
\begin{align}
 V_T={}&\frac12(m_S^2+c_ST^2)S^2+\frac{\lambda_S}{4}S^4\nonumber\\
 &+\frac12(-\mu_R^2+c_RT^2)R^2+\frac{\lambda_R}{4}R^4\nonumber\\
 &+\frac12(-\mu_H^2+c_HT^2)h^2+\frac{\lambda_H}{4}h^4 \nonumber\\
 &+\frac12(m_X^2+c_XT^2)X^2+\frac{\lambda_X}{4}(X^2)^2\nonumber\\
 &+\frac{\kappa_{SR}}{4}S^2R^2+\frac{\kappa_{SH}}{4}S^2h^2
 -\frac{\kappa_{SX}}{4}S^2X^2 .
 \label{eq:VT}
\end{align}
\rev{The polynomial approximation includes the light fermions perturbatively. Near $T_0$, however, the $SU(3)_\ell$ sector is assumed to be confined, so its scalar contribution is not generally captured by that expansion \cite{Pirzada:2026uak}. A chiral condensate contributes directly to the free energy \cite{Croon2020}. The corresponding term is linear in the absolute value of the light-quark mass near the chiral limit, in contrast to the analytic mass terms retained in the perturbative polynomial potential. For two degenerate light flavors, minimization over the chiral orientation and expansion in the quark mass gives}
\begin{equation}
 \rev{V_{\chi}(R,T)=-2y_p\Sigma_\ell(T)|R|+\mathcal O((y_pR)^2).}
 \label{eq:chiralfree}
\end{equation}
\fresh{The same chiral physics generates the low-sector susceptibility. For small positive $R$, Eq.~\eqref{eq:chilow} gives $\chi_\ell=y_p\Sigma_\ell R/2+\mathcal O(R^2)$. Its derivative and the scalar free-energy derivative are related by}
\begin{equation}
 \fresh{\left.\frac{\partial V_\chi}{\partial R}\right|_{0^+}
 =-4\left.\frac{\partial\chi_\ell}{\partial R}\right|_{0^+}
 =-2y_p\Sigma_\ell(T).}
 \label{eq:chiralidentity}
\end{equation}
\fresh{The factor four follows from the two-flavor chiral relations and does not depend on the pure Yang--Mills susceptibility. Reducing $y_p$ decreases the induced scalar force and the emerging axion curvature in the same proportion. A stationary high-temperature configuration at $R=0$ therefore requires either a vanishing chiral condensate in that phase or another contribution to the scalar potential that cancels the condensate force. The existence of a polynomial extremum is insufficient when the omitted term has a nonzero one-sided derivative at that point.}
\fresh{The constraint already applies at $R$ values well below the late-phase expectation value. The two susceptibilities are equal at $R=7.146\times10^{-6}\,\mathrm{GeV}$, where the chiral force already exceeds the polynomial restoring force by many orders of magnitude. Figure~\ref{fig:chiral} compares these contributions. At the chirally broken reference point, $R=0$ is not a stationary point of the full high-temperature free energy. We retain the polynomial interpolation to study conditional axion release, but a physical transition requires the chiral term to be included when the scalar vacua are determined.}
\begin{figure}[t]
\centering
\includegraphics[width=0.9\columnwidth]{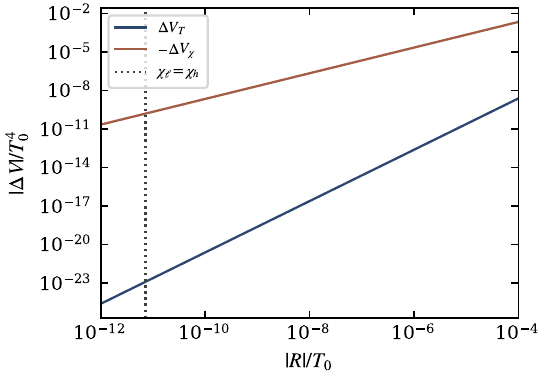}
\caption{\fresh{Chiral-condensate contribution to the scalar free energy and the positive polynomial cost at small $R$. The vertical line marks $\chi_\ell=\chi_h$. The condensate contribution dominates the scalar force in the range where the light-quark mass expansion applies.}}
\label{fig:chiral}
\end{figure}
For negative $c_S$, the favored high-temperature solution of the polynomial potential has
\begin{equation}
 S_>^2(T)=\frac{-m_S^2-c_ST^2}{\lambda_S},\qquad R_>=h_>=0,
 \label{eq:highbranch}
\end{equation}
when the numerator is positive. The competing low-temperature solution has
\begin{align}
 S_<&=0,\nonumber\\
 R_<^2(T)&=\frac{\mu_R^2-c_RT^2}{\lambda_R},\qquad
 h_<^2(T)=\frac{\mu_H^2-c_HT^2}{\lambda_H}.
 \label{eq:lowbranch}
\end{align}
The mass parameters are chosen so that the two polynomial branches are degenerate at $T_0=10^6\,\mathrm{GeV}$. Positive $S^2R^2$ and $S^2h^2$ couplings disfavor simultaneous nonzero expectation values and separate the extrema by a barrier.
For the numerical polynomial potential, we use
\begin{equation}
\begin{gathered}
 \lambda_S=\lambda_X=0.50,\quad \lambda_R=\lambda_H=0.05,\\
 \kappa_{SR}=1.50,\quad\kappa_{SH}=1.00,
 \quad \kappa_{SX}=0.44167,\\
 N_X=36,\quad y_q=0.25,\quad y_p=0.40311,\\
 \alpha_h=0.25,\quad m_S^2=0.20T_0^2,\\
 \mu_R^2=0.41180T_0^2,\quad \mu_H^2=0.80919T_0^2,\\
 m_X^2=0.50T_0^2.
\end{gathered}
\label{eq:couplings}
\end{equation}
They give
\begin{align}
 (c_S,c_R,c_H,c_X)&=(-0.700,0.300,0.697,1.547),\nonumber\\
 4\lambda_S\lambda_X-\kappa_{SX}^2&=0.805.
 \label{eq:thermalnumbers}
\end{align}
At specified $S$, minimization over $R$ and $h$ is analytic,
\begin{align}
 R^2(S,T)&=\left[\frac{\mu_R^2-c_RT^2-\kappa_{SR}S^2/2}{\lambda_R}\right]_+,
 \label{eq:Rreduce}\\
 h^2(S,T)&=\left[\frac{\mu_H^2-c_HT^2-\kappa_{SH}S^2/2}{\lambda_H}\right]_+,
 \label{eq:Hreduce}\\
 X^2(S,T)&=\left[\frac{\kappa_{SX}S^2/2-m_X^2-c_XT^2}{\lambda_X}\right]_+.
 \label{eq:Xreduce}
\end{align}
Minimization of the reduced polynomial potential gives the competing branches, with $X=0$ stable. For the spectator sector, the multiplicity-enhanced loop factor is $N_X\kappa_{SX}/(16\pi^2)=0.101$. Higher-order effects may thus affect the numerical precision of the thermal approximation. For the stationary scalar solutions, the condensate has a direct effect. If the low sector is chirally broken, its force along $R$ gives a nonzero derivative of the full potential at the high-$S$ point with $R=0$. This is more than a correction to the transition temperature: the point ceases to be an extremum of the full free energy. In a chirally symmetric phase, the linear term is absent; the corresponding confinement and susceptibility history has to be determined separately.
\begin{figure*}[t]
\centering
\includegraphics[width=0.9\textwidth]{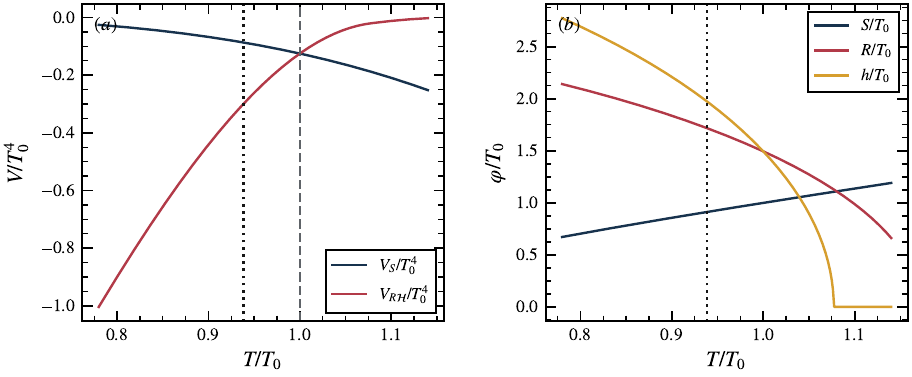}
\caption{Polynomial finite-temperature phase structure. (a) Free energies of the competing scalar branches, which cross at $T_0$; the dotted marker indicates the constrained-path reference $T_\star$. (b) Scalar expectation values on the high-$S$ and low-$(R,h)$ solutions. In panel (a), the branch with lower free energy is preferred within the polynomial approximation.}
\label{fig:transition}
\end{figure*}
\fresh{The polynomial extrema are separated by a barrier, so first-order decay is possible within that approximation. Its nucleation rate depends on the scalar trajectory through field space. We evaluate a surface tension along an imposed interpolation.} Thermal nucleation is governed by the $O(3)$ Euclidean bounce action, including the fluctuation prefactor \cite{Coleman1977,CallanColeman1977,Linde1983}. On the adopted interpolation path, the thin-wall expression is
\begin{equation}
 \frac{S_3^{\rm cp}}{T}=\frac{16\pi\sigma_{\rm cp}^3}{3T\Delta V^2},
 \qquad
 \sigma_{\rm cp}=\int_{\rm path}\dd\ell\sqrt{2[V_T(\ell)-V_>]}. 
 \label{eq:thinwall}
\end{equation}
\fresh{The thin-wall approximation applies only when the bubble radius is much larger than the thickness of its wall. Applying the reference criterion $S_3^{\rm cp}/T=140$ gives}
\begin{equation}
 T_\star=9.384\times10^5\,\mathrm{GeV},\quad
 \frac{\beta_{\rm cp}}{H_\star}=4.72\times10^3,
 \quad \alpha_{\rm PT}=0.127 .
 \label{eq:transitionnumbers}
\end{equation}
\fresh{The radius obtained is only $3.653$ times the wall thickness. The calculated radius is thus too small for a reliable determination of the nucleation temperature in the thin-wall limit. We examine the release dynamics over $T_\star\le T\le T_0$ instead. On the prescribed polynomial trajectory, the results are}
\begin{align}
 \mathcal N_0&\ge259.57,\qquad
 \left|\frac{A_f}{\delta_{\rm DM}}-1\right|\le3.82\times10^{-5},\nonumber\\
 \rev{0.1327}&\le\Delta_{S_a}^2(k_\star)\le\rev{0.1623}.
 \label{eq:transitionenvelope}
\end{align}
The envelope also gives $m_a\tau_w\le9.64\times10^{-3}$. \fresh{Throughout the interval, the initial correlations undergo strong damping, and capture proceeds on a short time scale. The quoted values belong to the prescribed path. Before they can be identified with a physical nucleation history, the scalar vacua must satisfy the chiral-force condition in Eq.~\eqref{eq:chiralidentity}.}
\section{Microscopic susceptibility and dissipation}\label{sec:rates}
The axion susceptibility determines the restoring curvature, while the dissipative part of the topological correlator governs both damping and thermal fluctuations. We calculate these quantities on either side of the prescribed scalar transition. Thermal friction acting on coherent axion motion has been studied previously \cite{UnruhWald1985,Turner1985,Papageorgiou2023,Choi2023,Banerjee2026}. Here the quantity of interest is the fluctuation population left after the curvature and coupling to the bath evolve. In the high-temperature phase, the masses $m_q=y_q|S_>|$ allow a nonzero $SU(2)_h$ susceptibility and hence an early axion curvature. Alignment is controlled by damping, which probes the real-time response of that gauge plasma rather than its static susceptibility. When $S=0$, both $q_i$ flavors are massless and the anomalous chiral Ward identity eliminates dependence on the high-sector vacuum angle. In the low-temperature phase, $m_p=y_p|R_<|$ induces an $SU(3)_\ell$ susceptibility, while the expectation value $h_<$ reduces $SU(2)_h$ Chern--Simons diffusion. \fresh{The polynomial interpolation specifies how these changes occur in our conditional release calculation; it does not solve the complete scalar-vacuum problem.}
The $SU(2)_h$ sector generates a nonzero axion potential when high-phase $q_i$ quarks have masses. We estimate its curvature from the one-loop finite-temperature instanton density \cite{GrossPisarskiYaffe1981}, with $b_0=6$, $A=2$, and $C_2=0.466\e^{-3.358}$ in our normalization,
\begin{align}
 \chi_h(T,S)={}&2C_2\int_0^{\rho_{\max}}\dd\rho\,\rho^{-5}
 [b_0\ln(1/\rho\Lambda_h)]^{4}\nonumber\\
 &\times\e^{-b_0\ln(1/\rho\Lambda_h)-A(\pi\rho T)^2}
 (y_q|S|\rho)^2 .
 \label{eq:DIGA}
\end{align}
Because the instanton amplitude contains two quark-mass insertions, $\chi_h\propto S^2$ and the temporary curvature goes to zero with $S$. This follows from the chiral properties of the light fermions, independently of whether the dilute-instanton expression remains accurate near the transition. At $T_\star$, numerical integration gives
\begin{align}
 \chi_h&=3.889\times10^{13}\,\mathrm{GeV}^4,\nonumber\\
 m_h&=\frac{\sqrt{\chi_h}}{f_a}=6.236\times10^{-3}\,\mathrm{GeV}.
 \label{eq:chihnumber}
\end{align}
Numerical integration using $\ln(\rho T)$ as the variable agrees at the relative level $5.9\times10^{-11}$. This is a quadrature check only; it does not establish the accuracy of the dilute-instanton approximation for the chosen coupling.
At the reference transition, we take the low gauge sector to be confined. Its vacuum angle is unobservable with massless quarks, even though a chiral condensate may still be present. A nonzero $R$ gives the quarks masses and restores the susceptibility. We use the chiral inverse-susceptibility relation \cite{DiVecchiaVeneziano1980,LeutwylerSmilga1992},
\begin{equation}
 \frac{1}{\chi_\ell(R)}=\frac{1}{\chi_{\rm YM}}+\frac{2}{y_p|R|\Sigma_\ell},
 \qquad \chi_{\rm YM}=\Lambda_\ell^4,
 \quad \Sigma_\ell=\Lambda_\ell^3.
 \label{eq:chilow}
\end{equation}
We set $\Lambda_\ell=3T_0$, corresponding to $T_\star/\Lambda_\ell=0.313$. Thus the low sector is confined across the reference transition, but its two massless $p_i$ quarks remove vacuum-angle dependence in the high phase. The susceptibility becomes nonzero when $R$ gives them mass. \fresh{The associated condensate force also acts on $R$, according to Eq.~\eqref{eq:chiralidentity}.} Confinement can therefore precede axion capture, although the mass generation and scalar equilibrium conditions remain coupled. Equation~\eqref{eq:chilow} gives $m_a(T_\star)=2.895\,\mathrm{TeV}$ at capture and $m_a(0)=3.618\,\mathrm{TeV}$ at zero temperature.
Thermal fluctuation--dissipation relates the retarded topological correlator to its symmetric counterpart. The frictional interaction is accordingly accompanied by a fluctuating force. In the local, low-frequency limit, we write
\begin{equation}
 \Gamma_{\rm CS}^{\rm sym}=\kappa_{\rm CS}\alpha_h^5T^4,
 \qquad
 \Upsilon=\frac{\Gamma_{\rm CS}}{2Tf_a^2}.
 \label{eq:friction}
\end{equation}
We normalize the rate at the scale indicated by calculations of non-Abelian sphaleron diffusion \cite{MooreTassler2011,DOnofrio2014,GuinSharma2026}. Those homogeneous calculations do not determine the correlator at finite momentum and frequency \cite{Bouzoud2026}, which is required for the integrated relic-energy estimate. Accordingly, the local treatment here represents an assumed response model, not a determination of each mode-dependent relaxation rate. With $\kappa_{\rm CS}=40$, we vary the diffusion rate from $0.3$ to $3$ times the benchmark, applying the same variation to friction and to the noise. Related thermal systems also exhibit scalar damping that depends on the mass \cite{Banerjee2026}. For the reference parameters,
\begin{equation}
 \Upsilon_-=1.742\times10^{-2}\,\mathrm{GeV},
 \qquad \frac{3H+\Upsilon_-}{H}=1.329\times10^4.
 \label{eq:frictionnumber}
\end{equation}
When $\mathcal H$ acquires an expectation value, topological transitions are suppressed by the barrier energy proportional to $h/g_h$. We estimate this effect with the broken-phase expression \cite{KlinkhamerManton1984,ArnoldMcLerran1987}
\begin{equation}
 \frac{\Gamma_{\rm CS}^{\rm br}}{\Gamma_{\rm CS}^{\rm sym}}
 \simeq\exp\!\left[-\frac{4\pi B h}{g_hT}\right],
 \qquad B=1.56,
 \label{eq:sphsupp}
\end{equation}
At the low-phase endpoint this suppression is $7.74\times10^{-11}$. The vanishing high-sector susceptibility and the reduction of friction have different origins. Massless $q_i$ remove the vacuum-angle dependence through the anomalous chiral Ward identity; the gauge condensate suppresses Chern--Simons transitions. Both effects take place during release on the imposed scalar path.
As the two susceptibilities exchange dominance, both the axion curvature and its preferred angle change. We describe this evolution by the wall potential
\begin{equation}
 V_a(\theta,z)=-\chi_h(z)\cos\theta
 -\chi_\ell(z)\cos(\theta-\delta),
 \label{eq:wallpotential}
\end{equation}
with harmonic curvature and moving minimum
\begin{equation}
 \omega_0^2(z)=\frac{\chi_h(z)+\chi_\ell(z)}{f_a^2},
 \qquad
 \theta_\star(z)=\delta\frac{\chi_\ell(z)}{\chi_h(z)+\chi_\ell(z)}.
 \label{eq:wallminimum}
\end{equation}
Figure~\ref{fig:switch} shows that the low-sector susceptibility becomes larger than the high-sector susceptibility before the scalar interpolation reaches its midpoint. The axion minimum begins to move in the early wall tail; the subsequent suppression of Chern--Simons diffusion reduces thermal forcing. \rev{There is consequently an interval in which the minimum is changing while the bath still excites the field. The overlap between these processes affects the covariance retained at release.}
\begin{figure*}[t]
\centering
\includegraphics[width=0.9\textwidth]{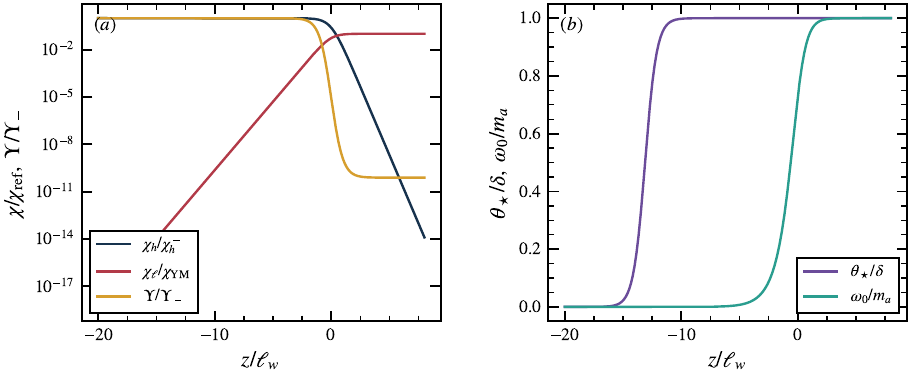}
\caption{\rev{Susceptibilities and dissipation along the imposed scalar profiles in Eq.~\eqref{eq:wallprofiles}.} (a) The high-sector curvature and Chern--Simons rate decrease while the low-sector susceptibility approaches its confined limit. (b) \rev{Enlarged early-wall region showing the crossing $\chi_\ell=\chi_h$ and the accompanying motion of the axion minimum.} The curves follow from Eqs.~\eqref{eq:DIGA}, \eqref{eq:chilow}, \eqref{eq:sphsupp}, and \eqref{eq:wallminimum}.}
\label{fig:switch}
\end{figure*}
\section{Thermal covariance and the noise bound}\label{sec:theorem}
The bath attenuates existing field fluctuations and simultaneously excites new ones. These effects are not separated by the mean-field equation. We instead follow the joint field--velocity covariance and distinguish the transported initial correlations from the positive contribution deposited by thermal noise. \fresh{For linear stochastic dynamics, this separation gives a lower bound on the energy supplied over a specified interval before late-time oscillations. The lower bound relies on positive-semidefinite noise covariance; equilibration to a stationary thermal state is not assumed.}
Kinetic fluctuations contribute to the released energy even when the field displacement variance is small. For each comoving mode, we introduce $\bm z_{\bm k}=(\vartheta_{\bm k},\dot\vartheta_{\bm k})^T$, where $\vartheta=\theta-\langle\theta\rangle$ and $\Gamma=3H+\Upsilon$. The local inertial Langevin equations are
\begin{align}
 \dot{\bm z}_{\bm k}&=A_k\bm z_{\bm k}+\bm\xi_{\bm k},\nonumber\\
 A_k&=\begin{pmatrix}0&1\\-\Omega_k^2&-\Gamma\end{pmatrix},
 \qquad \Omega_k^2=m^2+\frac{k^2}{a^2}.
 \label{eq:matrixLangevin}
\end{align}
The stochastic force enters the velocity equation, with its strength determined by the same thermal interaction as the friction. The fluctuation--dissipation relation is \cite{Miyamoto2014,Bartrum2015,CaoBoyanovsky2022,Ota2024}
\begin{align}
 \langle\bm\xi_{\bm k}(t)\bm\xi_{\bm k'}^T(t')\rangle
 ={}&(2\pi)^3\delta^3(\bm k+\bm k')N_k\delta(t-t'),\nonumber\\
 N_k={}&\begin{pmatrix}0&0\\0&2\Upsilon T/(a^3f_a^2)\end{pmatrix}.
 \label{eq:noiseMatrix}
\end{align}
Writing $C_k=\langle\bm z_{\bm k}\bm z_{-\bm k}^T\rangle$, we find that drift propagates existing correlations while the noise matrix adds a positive contribution. Their combined equation is
\begin{equation}
 \dot C_k=A_kC_k+C_kA_k^T+N_k.
 \label{eq:Lyapunov}
\end{equation}
For coefficients held constant during the interval, the stationary covariance is
\begin{equation}
 \Sigma_k=\frac{rT}{a^3f_a^2}
 \begin{pmatrix}\Omega_k^{-2}&0\\0&1\end{pmatrix},
 \qquad r=\frac{\Upsilon}{3H+\Upsilon},
 \label{eq:Sigma}
\end{equation}
which satisfies $A_k\Sigma_k+\Sigma_kA_k^T+N_k=0$.
\begin{proposition}[Full phase space preparation bound]\label{thm:full}
Let $U_k(\Delta t)=\exp(A_k\Delta t)$ and define
\begin{equation}
 \varepsilon_{2,k}=\left\|\Sigma_k^{-1/2}U_k\Sigma_k^{1/2}\right\|_2^2,
 \qquad 0\le\varepsilon_{2,k}\le1.
 \label{eq:epsilon2}
\end{equation}
For any positive initial covariance $C_{k,i}$,
\begin{align}
 C_{k,f}&=U_kC_{k,i}U_k^T+W_k,\label{eq:Csolution}\\
 W_k&=\Sigma_k-U_k\Sigma_kU_k^T
 \succeq(1-\varepsilon_{2,k})\Sigma_k .
 \label{eq:matrixbound}
\end{align}
Consequently,
\begin{equation}
 C_{k,f}\succeq(1-\varepsilon_{2,k})\Sigma_k,
 \qquad
 \det C_{k,f}\ge(1-\varepsilon_{2,k})^2\det\Sigma_k.
 \label{eq:detbound}
\end{equation}
Two initial covariances evolved with identical coefficients differ by $U_k(C_{k,i}^{(1)}-C_{k,i}^{(2)})U_k^T$.
\end{proposition}
\begin{proof}[Derivation]
Variation of constants gives Eq.~\eqref{eq:Csolution}, where
$W_k=\int_0^{\Delta t}U_k(s)N_kU_k^T(s)\dd s\succeq0$. Integrating the stationary Lyapunov equation gives $W_k=\Sigma_k-U_k\Sigma_kU_k^T$, so $U_k\Sigma_kU_k^T\preceq\Sigma_k$ and $0\le\varepsilon_{2,k}\le1$. The operator norm definition in Eq.~\eqref{eq:epsilon2} implies
$\Sigma_k^{-1/2}U_k\Sigma_kU_k^T\Sigma_k^{-1/2}\preceq\varepsilon_{2,k}I$, establishing Eq.~\eqref{eq:matrixbound}. Since $U_kC_{k,i}U_k^T$ is positive semidefinite, the first relation in Eq.~\eqref{eq:detbound} follows. Monotonicity of the determinant for positive matrices gives the second. Subtracting the two covariance solutions proves the final statement.
\end{proof}
The bound applies to the displacement, velocity, and mixed covariances. The open-system treatment of Ref.~\cite{CaoBoyanovsky2022} contains thermal particle production alongside coherent damping; the matrix relation here quantifies how much of the stationary phase-space covariance is supplied by noise over a finite interval. The quantity $\varepsilon_{2,k}$ measures the largest surviving initial contribution in the metric set by $\Sigma_k$. When it is small, thermal forcing has supplied nearly the stationary covariance in both quadratures. \rev{The evolution is an affine map with a positive-semidefinite noise contribution.} Equation~\eqref{eq:detbound} also provides a lower limit on the classical Gaussian phase-space volume and, relative to a specified reference cell, its differential entropy. Loss of initial correlations is therefore compatible with a finite final variance and energy. It is the separate propagation of the incoming state and the accumulated noise that permits a statement independent of the amount of initial coherent or stochastic preparation. These are supplied by the bath rather than by the original displacement.
\subsection{Spatial spectrum and thermodynamic identity}
The covariance determines the energy in fluctuations, while its dependence on momentum determines which modes carry that energy. In the overdamped regime $\Omega_k/\Gamma\ll1$, elimination of the rapidly relaxing velocity gives
\begin{equation}
 \dot\vartheta_{\bm k}+\lambda_k\vartheta_{\bm k}=\eta_{\bm k},
 \qquad \lambda_k=\frac{\Omega_k^2}{\Gamma},
 \label{eq:OD}
\end{equation}
with diffusion coefficient $D=\Upsilon T/(a^3f_a^2\Gamma^2)$. The field power satisfies
\begin{equation}
 \dot P_\theta=-2\lambda_kP_\theta+2D.
 \label{eq:PspectrumODE}
\end{equation}
Noise acting at different times contributes positively to the final spectrum. For positive time-dependent coefficients, the thermal part of the field power can be written as the Laplace integral
\begin{equation}
 P_{\rm th}(k)=2\int_{t_i}^{t_f}\dd t\,D(t)
 \exp[-2A(t)-2k^2B(t)],
 \label{eq:LaplaceMain}
\end{equation}
where $A(t)=\int_t^{t_f}m^2/\Gamma\,\dd u$ and $B(t)=\int_t^{t_f}(a^2\Gamma)^{-1}\dd u$. Hence
\begin{equation}
 (-1)^n\partial_{(k^2)}^nP_{\rm th}(k)\ge0,
 \qquad
 \Delta_\theta^2(k)=\frac{P_0}{2\pi^2}k^3+\mathcal O(k^5).
 \label{eq:completeMono}
\end{equation}
The positive kernel implies that, in this overdamped local harmonic limit, the field power is completely monotone in $k^2$. The dimensionless infrared spectrum has the $k^3$ behavior characteristic of finite-range correlations \cite{HohenbergHalperin1977,Kibble1976,Zurek1985,Chandran2012}, while Eq.~\eqref{eq:completeMono} also constrains its momentum dependence away from $k=0$. Oscillations or alternating spectral shoulders are absent in this approximation. Such features can arise only after including physics excluded from this reduction, for example inertia, spatially nonlocal response, or nonlinear coupling. Complete monotonicity consequently characterizes this particular approximation and is not a general restriction on nonequilibrium axion spectra.
In a quasistatic thermal interval, Eq.~\eqref{eq:Sigma} gives
\begin{equation}
 P_\theta(k)=\frac{rT}{a^3f_a^2(m^2+k^2/a^2)}.
 \label{eq:staticP}
\end{equation}
With $k_\star=am$ and $P_0=P_\theta(0)$, the stationary spectrum obeys the identity
\begin{equation}
 P_0k_\star^2=\frac{rT}{af_a^2}.
 \label{eq:susidentity}
\end{equation}
A larger curvature shortens the correlation length, shifts the turnover toward higher $k_\star$, and lowers the infrared power. At fixed $rT/(af_a^2)$, the product of the zero-momentum amplitude and the square of the turnover scale does not change. The relaxation rate sets the approach to this stationary spectrum, while \rev{the amplitude--scale combination retains the thermal factor $r$.} The analogous dilation relation for homogeneous power-law release is derived in Appendix~\ref{app:powerlaw}: $P_0\widehat k^2=I_p rT/(af_a^2)$. \rev{For fixed $r$, the mass drops out explicitly; the effect of friction remains contained in $r=\Upsilon/(3H+\Upsilon)$.}
\section{Transition, release, and late capture}\label{sec:transition}
After plasma decoupling, the axion oscillates in the potential of the confined gauge sector. Its energy depends on the released covariance as well as the remaining coherent field. When the minimum changes on a time scale shorter than an oscillation period, both contributions pass through the matching. Sudden changes can retain coherent excitation \cite{Khan:2026nsz}. Here we also evolve the covariance generated by the thermal medium and calculate the energy associated with it. The field state is continuous across release, but its oscillator energy is evaluated with the curvature of the final sector.
\rev{We connect the polynomial extrema with a prescribed interpolation.} Writing $u(z)=[1+\tanh(z/\ell_w)]/2$,
\begin{equation}
 S=S_>(1-u),\qquad R=R_<u,\qquad h=h_<u,
 \label{eq:wallprofiles}
\end{equation}
we insert these fields into the susceptibility and damping functions of Sec.~\ref{sec:rates}. The mean and the full inertial covariance then satisfy
\begin{align}
 \ddot{\bar\theta}+\Gamma(z)\dot{\bar\theta}
 +\omega_0^2(z)[\bar\theta-\theta_\star(z)]&=0,
 \label{eq:meanwall}\\
 \dot C_k=A_k(z)C_k+C_kA_k^T(z)+N_k(z).&
 \label{eq:covwall}
\end{align}
\rev{We initialize the covariance with its high-phase stationary value where $\chi_\ell/\chi_h$ is small. For the homogeneous mode, we separately evolve the state during cooling to test this choice.} During release, $\chi_h$ and $\Upsilon$ decrease while $\chi_\ell$ approaches its low-phase value. As the curvature evolves, the stochastic force decreases. The covariance carried into the final oscillations depends on the timing of these two changes.
The reference solution is overdamped before it encounters the wall. The smaller decay eigenvalue of the inertial drift matrix is
\begin{equation}
 \lambda_-(T)=\frac{\Gamma(T)-\sqrt{\Gamma^2(T)-4m_h^2(T)}}{2},
 \qquad \Gamma=3H+\Upsilon .
 \label{eq:lambdaslow}
\end{equation}
The integrated alignment depth during radiation-dominated cooling is
\begin{equation}
 \Ncal_0=\int_{T_\star}^{1.5T_0}\frac{\lambda_-(T)}{H(T)T}\dd T=368.29,
 \label{eq:depthnumber}
\end{equation}
The calculated depth corresponds to strong relaxation of initial correlations over this cooling interval. Different initial covariances approach practically the same pre-wall state, with appreciable alignment also found when susceptibility and friction are varied separately. Their convergence measures thermal relaxation rather than depletion of the final fluctuations. A nearly unique covariance after cooling can therefore coexist with a pronounced dependence of the final abundance on the number of thermally populated modes included in the energy integral. Noise replaces the incoming covariance with a thermally prepared one, so the attenuation of the homogeneous axion does not determine the energy in the excited modes.
\rev{We normalize the coherent oscillator by first neglecting the stochastic component and matching its energy to the observed dark matter abundance. The resulting displacement provides a reference for comparing thermal production. Allowing for adiabatic conservation of particle number as the mass changes from $m_a(T_\star)$ to $m_a(0)$ gives}
\begin{equation}
 \frac{\rho_a^{\rm coh}}{s}=\frac{m_a(0)m_a(T_\star)f_a^2\delta_{\rm DM}^2}{2s(T_\star)}.
 \label{eq:abundance}
\end{equation}
\rev{For $g_{*s}=120$ and $\rho_{\rm DM}/s=0.44\,\mathrm{eV}$, this coherent-only normalization gives}
\begin{equation}
 \delta_{\rm DM}=6.044\times10^{-8}.
 \label{eq:deltanumber}
\end{equation}
The corresponding angle is small enough for the coherent motion to be harmonic. The thermal covariance adds positive energy, however, and changes the displacement required when the total relic density is normalized.
In the reference trajectory, the wall passage time is short compared with the late axion period, $m_a(T_\star)\tau_w\ll1$. Sudden matching is then a useful approximation for the mean and the two fluctuation quadratures. We define the phase-space spectrum
\begin{equation}
 P_A(k)=C_{\theta\theta}(k)+
 \frac{C_{\dot\theta\dot\theta}(k)}{m_a^2+k^2/a^2},
 \label{eq:PA}
\end{equation}
After averaging over oscillator phase, the field--velocity cross term does not enter the mean energy, although it remains part of the Gaussian state. Two states with equal displacement power can carry different energies if their velocity variances differ. The relic-energy calculation consequently requires both diagonal covariances. \fresh{For the prescribed wall, direct propagation and sudden matching give consistent retained fluctuation energies.}
\fresh{A small mean displacement after alignment does not imply a small axion relic. Both displacement and velocity fluctuations supply positive energy when the late oscillator begins to evolve. With $\omega_{a,p}^2=m_a^2+p^2$, the energy integrated up to physical momentum $K$ is}
\begin{equation}
 \rev{\rho_{\rm th}^{(<K)}=\frac{f_a^2}{2}\int_{|\bm p|<K}\frac{\dd^3p}{(2\pi)^3}\left[\omega_{a,p}^2P_{\theta}(p)+P_{\dot\theta}(p)\right].}
 \label{eq:thermalenergy}
\end{equation}
\fresh{As a first comparison, we evaluate the mode integral using the stationary local covariance of Eq.~\eqref{eq:Sigma}. This gives the finite-band result}
\begin{align}
 \rev{\frac{\rho_{\rm th}^{(<K)}}{\rho_{\rm coh}}=}&
 \rev{\frac{rT_\star}{2\pi^2 f_a^2m_a^2\delta_{\rm DM}^2}}
 \nonumber\\
 &\rev{\times\left[(m_a^2-m_h^2)\left(K-m_h\arctan\frac{K}{m_h}\right)
 +\frac{2K^3}{3}\right].}
 \label{eq:finiteband}
\end{align}
\fresh{For $K=\Gamma$, this stationary contribution is $1.271\times10^{-1}$ of the coherent-only reference energy. It reaches the same energy at $K=8.621\times10^{-2}\,\mathrm{GeV}$. In the interval $m_h\ll p\ll m_a$, the $p^2$ momentum-shell factor approximately compensates the $P_\theta\propto p^{-2}$ field spectrum. A modest amplitude near the spectral turnover can therefore correspond to an appreciable momentum-integrated energy. We next obtain a lower estimate without assigning an equilibrium covariance at the start of the interval.}
\begin{equation}
 \rev{\frac{\rho_{\rm coh}+\rho_{\rm th}}{s}=\frac{\rho_{\rm DM}}{s},\qquad
 \rho_{\rm th}\geq\rho_{\rm th}^{(<K)}\geq0.}
 \label{eq:fullabundance}
\end{equation}
\fresh{Thermal forcing produces covariance even if no fluctuations are present initially. Consider only the final cooling interval, from $t_i$ to the formation of the wall at $t_\star$. The covariance supplied by noise over this interval is}
\begin{equation}
 \fresh{W_k^{\rm late}=\int_{t_i}^{t_\star}U_k(t_\star,t)N_k(t)
 U_k^T(t_\star,t)\,\dd t\succeq0.}
 \label{eq:lateGramian}
\end{equation}
\fresh{Any covariance already present at $t_i$ gives a further positive contribution. Thus the covariance entering the wall is bounded below by $W_k^{\rm late}$. Propagation with the wall evolution matrix $M_{w,k}$ preserves positive-semidefinite ordering, and noise acting within the wall can only add covariance. It follows that}
\begin{align}
 \fresh{C_{k,\rm out}}&\fresh{\succeq M_{w,k}W_k^{\rm late}M_{w,k}^T,}\nonumber\\
 \fresh{\rho_{\rm th}^{(<K)}}&\fresh{\ge\frac{f_a^2}{2}\int_{|\bm p|<K}\frac{\dd^3p}{(2\pi)^3}
 \Tr\!\left[D_{a,p}M_{w,p}W_p^{\rm late}M_{w,p}^T\right],}\nonumber\\
 \fresh{D_{a,p}}&\fresh{=\diag(m_a^2+p^2,1).}
 \label{eq:energyfloor}
\end{align}
\fresh{To evaluate this lower contribution, we set the covariance to zero at $T_i=1.01T_\star$, integrate the local Langevin system through radiation-era cooling, and then propagate it across the imposed wall without adding wall noise. The result contains only the fluctuations produced in the final one percent of high-phase cooling. In particular, no earlier equilibrium population is included, and omitting additional stochastic forcing within the wall makes the resulting contribution conservative within the adopted local Langevin dynamics. The evolution accounts for redshifting of the momentum and the $a^{-3}$ scaling of comoving noise, taking $a(T_\star)=1$.}
\fresh{In Fig.~\ref{fig:energy}, we plot the energy produced during this cooling interval. The thermal and coherent-only contributions become equal at $K=8.619\times10^{-2}\,\mathrm{GeV}$, well below the temperature of the plasma. The equality identifies the momentum range relevant to the local calculation, without establishing the form of the dissipative kernel at those momenta. The required correlator at nonzero momentum cannot be extracted from homogeneous Chern--Simons diffusion alone \cite{Bouzoud2026,GuinSharma2026}. Within the local approximation, the energy in these modes reaches the coherent benchmark and is not eliminated by choosing a smaller initial angle. To obtain a physical abundance one needs the gauge correlator over the integration range rather than a freely imposed cutoff. The quoted momentum marks equality of the two contributions in the local calculation. It is not a derived cutoff or a physical saturation momentum for the relic density.}
\begin{figure*}[t]
\centering
\includegraphics[width=0.9\textwidth]{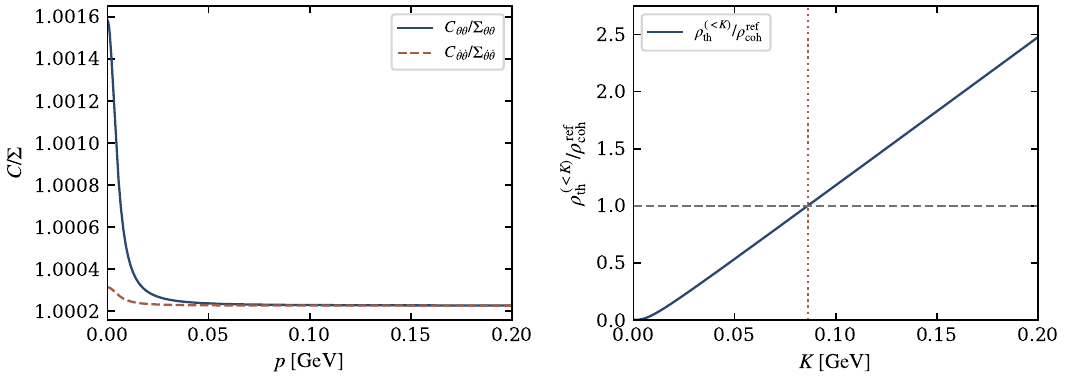}
\caption{\fresh{Fluctuation energy generated during the final one percent of high-phase cooling and propagated through the prescribed wall without additional wall noise. The left panel compares the evolved covariance with its stationary local reference. The right panel shows the momentum-integrated lower contribution relative to the coherent-only energy. The comparison depends on the full integration band, not only on the power at the turnover.}}
\label{fig:energy}
\end{figure*}
\rev{For an energy density dominated by the coherent oscillator, the leading density perturbation depends linearly on the axion fluctuation. In this approximation, the isocurvature power takes the form}
\begin{equation}
 \Delta_{S_a}^2(k)=\frac{4}{\delta_{\rm DM}^2}
 \frac{k^3P_A(k)}{2\pi^2}.
 \label{eq:isocurvature}
\end{equation}
At $k_\star=am_h$,
\begin{equation}
 \Delta_{S_a}^2(k_\star)=1.623\times10^{-1},
 \qquad
 k_\star=7.79\times10^{16}\,\mathrm{Mpc}^{-1}.
 \label{eq:isovalues}
\end{equation}
\rev{The modes at the spectral turnover are strongly nonrelativistic,} $k_\star^2/m_a^2=4.64\times10^{-12}$. This scale is associated with microscopic plasma correlations, rather than with a primordial superhorizon perturbation. The reference linear contrast is below unity at $k_\star$. \rev{When the fluctuating component carries a significant part of the mean energy, the density correlator contains terms quadratic in the field fluctuations. The energy normalization must account for this component in assigning a dark matter isocurvature perturbation. On sufficiently short wavelengths the reference contrast becomes of order unity, and a linear treatment of the density ceases to apply. The corresponding plotted spectrum is consequently a diagnostic of the coherent-dominated reference state, rather than a full prediction for the density field when the thermal component becomes substantial.}
\begin{figure*}[t]
\centering
\includegraphics[width=0.9\textwidth]{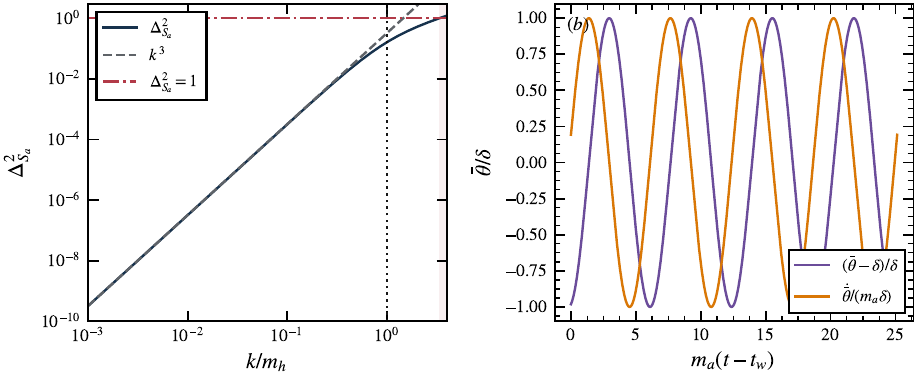}
\caption{Fluctuations released from the plasma and oscillations in the final axion potential. \rev{(a) Isocurvature spectrum in the coherent-only reference, calculated from the field and velocity covariances.} The dotted line marks $k_\star=m_h$, the dashed line the infrared $k^3$ dependence, and the dot-dashed line $\Delta_{S_a}^2=1$. The shaded region lies beyond the linear-density regime. (b) Mean field displacement and velocity during four periods after the wall, showing oscillations around the late minimum.}
\label{fig:capture}
\end{figure*}
\begin{table}[t]
\caption{Model parameters and preparation quantities for the benchmark. The value of $T_\star$ follows from the constrained interpolation in Eq.~\eqref{eq:thinwall}.}
\label{tab:benchmark}
\centering
\begin{tabular}{@{}lcc@{}}
\toprule
Quantity & Symbol & Value\\
\midrule
Transition benchmark & $T_\star$ & $9.384\times10^5\,\mathrm{GeV}$\\
High-sector curvature & $m_h$ & $6.236\times10^{-3}\,\mathrm{GeV}$\\
Friction to expansion & $\Gamma/H_\star$ & $1.329\times10^4$\\
Alignment depth & $\Ncal_0$ & $3.683\times10^2$\\
Late axion mass & $m_a(T_\star)$ & $2.895\,\mathrm{TeV}$\\
Reference coherent angle & $\delta_{\rm DM}$ & $6.044\times10^{-8}$\\
\bottomrule
\end{tabular}
\end{table}
\section{Dynamical stability and physical scales}\label{sec:tests}
The released state depends on the duration of plasma contact and on the rate at which the final potential appears. We vary the bath correlation time, susceptibility, Chern--Simons diffusion rate, and duration of the wall passage. We also compare the resulting microscopic fluctuation lengths with the scales probed by cosmological isocurvature measurements.
\subsection{Finite bath correlation time and dynamical stability}
\fresh{A bath with finite correlation time exerts a retarded dissipative force, and its noise is correlated at unequal times. The bath correlation time is denoted by $\tau_b$; it is independent of how the initial axion state is specified. Friction and noise are varied consistently. We take the exponential retarded kernel $\mathcal K(t)=\Upsilon\exp(-t/\tau_b)/\tau_b$ for $t>0$, together with $\langle\xi(t)\xi(t')\rangle=(\Upsilon T/f_a^2\tau_b)\exp(-|t-t'|/\tau_b)$. The generalized Langevin equation is}
\begin{equation}
 \fresh{\ddot\theta+3H\dot\theta+\omega_0^2\theta
 +\int_{-\infty}^{t}\mathcal K(t-t')\dot\theta(t')\,\dd t'=\xi(t).}
 \label{eq:GLEmatched}
\end{equation}
\begin{equation}
 \fresh{\tau_b\lambda_-=4.481\times10^{-8},\qquad
 \lambda_- = \frac{\Gamma-\sqrt{\Gamma^2-4m_h^2}}{2}.}
 \label{eq:markov}
\end{equation}
\fresh{This choice preserves the classical fluctuation--dissipation relation. In the absence of expansion, the stationary field and velocity variances obey equipartition for any $\tau_b$. For $\tau_b=(\alpha_h^2T_\star)^{-1}$, the slow relaxation rate including Hubble friction gives $\tau_b\lambda_-=4.481\times10^{-8}$. The stationary variances differ from their Markov values by less than $1.500\times10^{-10}$. Figure~\ref{fig:sensitivity}(a) shows significant departures only for correlation times much larger than this reference. For the benchmark homogeneous mode, these results justify the local approximation to the temporal bath response. This agreement tests temporal nonlocality for the selected mode, but it does not constrain the spatial nonlocality relevant to the finite-momentum energy integral.}
We also vary susceptibility and Chern--Simons diffusion independently in the $3\times3$ scan of Fig.~\ref{fig:sensitivity}(b). In all nine cases the alignment depth exceeds $11.25$. Changing the wall duration by factors $0.1$, $1$, and $10$ gives fractional changes of $3.2\times10^{-7}$, $3.2\times10^{-5}$, and $3.16\times10^{-3}$ in the captured mean amplitude. For a duration larger by $10^3$, the change is $0.731$. The benchmark wall is thus in the regime of rapid release, whereas a sufficiently slow passage allows the axion to adjust adiabatically. This scan concerns the assumed wall profile and does not determine nucleation from the full scalar potential.
\begin{figure*}[t]
\centering
\includegraphics[width=0.9\textwidth]{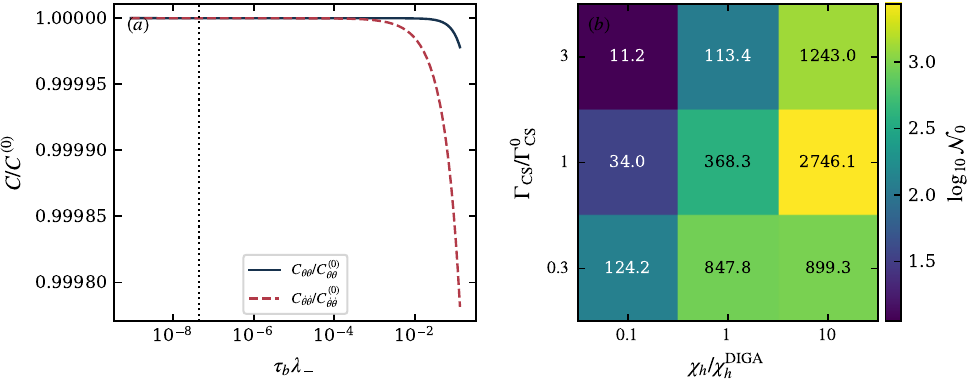}
\caption{\fresh{Sensitivity to bath and transition parameters. (a) With a finite bath correlation time and a fluctuation--dissipation-consistent noise kernel, the stationary covariance approaches its local limit near the reference response time. (b) The alignment remains appreciable under separate variations of the high-sector susceptibility and the Chern--Simons rate.}}
\label{fig:sensitivity}
\end{figure*}
\subsection{Cosmological scale dependence}
\fresh{Thermally produced axion fluctuations are correlated over microscopic distances, in contrast to perturbations extending beyond the horizon during inflation. The interpretation of nonstandard inflationary initial states also depends on the inflationary duration \cite{Khan:2026doo}. For finite spatial correlations, $P_\theta(k)$ remains regular as $k\to0$, giving $\Delta_\theta^2\propto k^3$ in the infrared. The turnover momentum is many orders of magnitude above the CMB range. A similar infrared spectrum occurs for causal fluctuations of finite correlation length \cite{Enander2017,Feix2020}. This infrared scaling alone does not specify the later matter distribution, particularly when velocity and gradient energies are important in the released axion population.}
By contrast, the inflationary axion-isocurvature spectrum is nearly scale invariant, with a characteristic amplitude set by $H_I/(\pi f_a\theta_i)$ \cite{Planck2018X,Marsh2016}. \rev{Entropy conservation gives $k_\star=7.79\times10^{16}\,\mathrm{Mpc}^{-1}$ for the coherent-only reference turnover. Extrapolating the spectrum to the CMB pivot gives}
\begin{equation}
 \Delta_{S_a}^2(0.05\,\mathrm{Mpc}^{-1})=4.285\times10^{-56}.
 \label{eq:CMBvalue}
\end{equation}
At the reference scale, the Planck limit for uncorrelated cold dark matter isocurvature is $\beta_{\rm iso}<0.038$ \cite{Planck2018X}. Using the Planck best fit scalar amplitude $A_s=2.10\times10^{-9}$ \cite{Planck2018VI}, the corresponding scale invariant reference power is
\begin{equation}
 \Delta_{S,{\rm inf}}^2=\frac{\beta_{\rm iso}}{1-\beta_{\rm iso}}A_s
 <8.30\times10^{-11}.
 \label{eq:Planckpower}
\end{equation}
\fresh{On CMB scales the calculated thermal contribution is very small, although the same population can carry considerable energy at higher momenta. Postinflationary finite-range axion fluctuations exhibit the same $k^3$ infrared scaling \cite{Enander2017,Feix2020}; here their amplitude is set by thermal noise and dissipation rather than a chosen distribution of initial angles. For an appreciable stochastic energy fraction, quadratic field terms contribute to the density spectrum as well \cite{Gurmani2026}. A calculation of small-scale structure then needs both the momentum distribution at release and the complete density correlator, rather than the linear field power alone.}
A first-order transition can also source gravitational radiation through fluid motion \cite{Caprini2020}. The spectrum depends on nucleation and bubble-wall dynamics, which are not fixed by the polynomial interpolation used here. We accordingly impose no sound-wave constraint on the benchmark.
\fresh{Homogeneous relaxation is governed by the $k=0$ response, but the relic-energy integral receives contributions from a finite range of momenta. To use the released covariance in a calculation of small-scale density evolution, the gauge correlator at finite momentum has to be determined.}
\section{Conclusion}\label{sec:conclusion}
We have studied the axion relic energy in the presence of thermal interactions with a non-Abelian plasma. Friction reduces the coherent axion amplitude, while plasma noise excites both the field and its velocity. Evolution of the field and velocity quadratures through cooling and the imposed scalar transition shows that the last interval of thermal forcing can contribute energy of the same order as the coherent-only dark matter reference. This energy comes from the positive covariance accumulated over that interval, without specifying an initially equilibrated state. The numerical comparison assumes a local bath response across the finite-momentum integration band. Neither the damping of the mean displacement nor the limited field power at the turnover is enough to infer a small integrated relic abundance.
An independent restriction comes from the scalar phase structure. The condensate responsible for the confined-sector axion susceptibility contributes a force along the scalar direction $R$. In the chirally broken reference phase, adding the chiral free energy removes stationarity of the polynomial high-temperature point at $R=0$. Reducing the light-quark Yukawa interaction decreases both the scalar force and the susceptibility generated as the quarks acquire mass. The imposed wall can still be used to calculate conditional axion dynamics, but at this reference point it does not connect stationary vacua of the full scalar free energy.
Thermally prepared fluctuations are correlated on microscopic scales, with dimensionless infrared field power scaling as $k^3$. At CMB wavelengths the extrapolated power is negligible, although the integrated energy need not be small. For shorter wavelengths, the density correlation function depends on the full momentum distribution and, when stochastic energy is important, on quadratic fluctuation terms. A microscopic treatment of this scenario requires a gauge response at finite momentum consistent with the energy transferred by the bath, together with scalar minima that exist for the assumed thermal history. The covariance inequality constrains the thermal contribution to the energy, while the chiral force imposes an independent condition on the scalar phases.
\section*{Acknowledgment}
ChatGPT was used only to edit the language of this manuscript, including its grammar, clarity, and readability. The authors are responsible for all scientific content, mathematical derivations, and the final version of the manuscript.

\appendix
\section{Power law release and dilation identity}\label{app:powerlaw}
For a restoring curvature that decreases as a power law, the fluctuation amplitude can be related analytically to its characteristic momentum. We use this limit to compare with the prescribed scalar-wall evolution. In the overdamped equation \eqref{eq:OD}, we introduce $s=t_c-t>0$ and take
\begin{equation}
 m^2(s)=\mu^2(s/\tau_Q)^p,
 \qquad p>0,
 \label{eq:powerlawA}
\end{equation}
where $a,T,\Gamma$, and $r$ are constant during the release. Define
\begin{equation}
 \widehat t=\left[\frac{(p+1)\Gamma\tau_Q^p}{2\mu^2}\right]^{1/(p+1)},
 \qquad
 \widehat k^2=\frac{a^2\Gamma}{2\widehat t}.
 \label{eq:hatsA}
\end{equation}
Using $u=s/\widehat t$ in Eq.~\eqref{eq:LaplaceMain} gives
\begin{align}
 P_\theta(k)&=P_0\Fcal_p(x),\qquad x=k^2/\widehat k^2,\label{eq:powerPA}\\
 \Fcal_p(x)&=\frac{1}{I_p}\int_0^\infty\e^{-u^{p+1}-xu}\dd u,\label{eq:FpA}\\
 I_p&=\frac{1}{p+1}\Gamma\!\left(\frac{1}{p+1}\right),\label{eq:IpA}\\
 P_0&=\frac{2rT\widehat t}{a^3f_a^2\Gamma}I_p.
 \label{eq:P0A}
\end{align}
Multiplication by $\widehat k^2$ yields
\begin{equation}
 P_0\widehat k^2=I_p\frac{rT}{af_a^2}.
 \label{eq:dilationA}
\end{equation}
A change in the quench time affects the turnover scale and amplitude oppositely, leaving their product invariant. Thus the homogeneous power-law result satisfies the dilation relation. When the turnover is defined using the instantaneous mass, the result reduces to the stationary susceptibility relation of Eq.~\eqref{eq:susidentity}.
\section{Time dependent covariance and overdamped bound}\label{app:timecov}
When the bath coefficients depend on time, the covariance evolves with no single stationary reference. The initial contribution and the accumulated noise can nevertheless be separated exactly. For $A_k(t)$ varying with time, the fundamental matrix $U_k(t,t')$ gives
\begin{align}
 C_k(t_f)={}&U_k(t_f,t_i)C_k(t_i)U_k^T(t_f,t_i)\nonumber\\
 &+\int_{t_i}^{t_f}U_k(t_f,t)N_k(t)U_k^T(t_f,t)\dd t.
 \label{eq:generalGramian}
\end{align}
Every propagated noise matrix, and hence its time integral, is positive semidefinite. Unlike the constant-coefficient case, a bound expressed relative to a stationary covariance also involves the evolution of that reference. In the scalar overdamped projection, let $D=\lambda_k\chi_k$ and define
\begin{equation}
 \Mcal_k=\exp\!\left[-2\int_{t_i}^{t_f}\lambda_k\dd t\right].
 \label{eq:memoryA}
\end{equation}
Direct integration of Eq.~\eqref{eq:PspectrumODE} gives
\begin{equation}
 P_f=\Mcal_kP_i+2\int_{t_i}^{t_f}D(t)
 \exp\!\left[-2\int_t^{t_f}\lambda_k\dd u\right]\dd t.
 \label{eq:scalarSolutionA}
\end{equation}
If $\chi_{k,-}\le\chi_k(t)\le\chi_{k,+}$ over the support of $\lambda_k$, then
\begin{equation}
 \chi_{k,-}(1-\Mcal_k)\le P_{\rm th}(k)
 \le\chi_{k,+}(1-\Mcal_k).
 \label{eq:scalarBoundA}
\end{equation}
The positive integration kernel has total weight $1-\Mcal_k$, which gives the two bounds.
\bibliography{refs}
\end{document}